\documentclass{article}
\usepackage{amssymb}

\usepackage{graphicx,authblk}
\usepackage{amsmath,mathrsfs}

\newtheorem{theorem}{Theorem}

\newtheorem{definition}[theorem]{Definition}

\newtheorem{proposition}[theorem]{Proposition}

\newenvironment{proof}[1][Proof]{\textbf{#1.} }{\ \rule{0.5em}{0.5em}} 

\begin{document}

\title{The Estimation Lie Algebra Associated with Quantum Filters}

\author[1]{Nina H. Amini}
\author[2]{John E. Gough}
\affil[1]{nina.amini@lss.supelec.fr, Laboratoire des Signaux et Syst\`{e}mes, 
CNRS - CentraleSup\'{e}lec - Univ. Paris-Sud, Universit\'{e} Paris-Saclay, 3, rue Joliot Curie, 91192 Gif-sur-Yvette, France }
\affil[2]{jug@aber.ac.uk, Department of Physics, Aberystwyth University, SY23 3BZ, Wales, United Kingdom}

\date{}

\maketitle
 
\begin{abstract}
We introduce the Lie algebra of super-operators associated with a quantum filter, specifically emerging from the Stratonovich calculus. In classical filtering, the analogue algebra leads to a geometric theory of nonlinear filtering which leads to well-known results by Brockett and by Mitter characterizing potential models where the curse-of-dimensionality may be avoided, and finite dimensional filters obtained. We discuss the quantum analogue to these results. In particular, we see that, in the case where all outputs are subjected to homodyne measurement, the Lie algebra of super-operators is isomorphic to a Lie algebra of system operators from which one may approach the question of the existence of finite-dimensional filters.
\end{abstract}

\section{Introduction}
The aim of this paper is to extend to the quantum domain some fruitful approaches to the problem of gauging the complexity of the estimation problem coming from classical filtering theory.
We recall the basic problem in nonlinear filtering. A system has state variable $X(t)$ on $\mathbb{R}^n$, say, which undergoes a noisy evolution, and we wish to obtain the conditional expectation, $\pi_t (f)$, of $f(X(t))$ based on noisy observations of a process $Y$ up to time $t$, for arbitrary twice differentiable bounded functions $f$. Let us take the model equations to be
\begin{eqnarray*}
dX(t) &=& v(X(t)) \, dt + \gamma_0 dW(t) , \\
dY(t) &=& h(X(t)) \, dt + dZ(t) ,
\end{eqnarray*}
where $W(t)$ and $Z(t)$ are independent Wiener processes.
The filter may be written in the form $\pi_t (f) = \frac{\sigma_t (f)}{\sigma_t (1)}$ where $ \sigma_t (f) \equiv \int f(x) \sigma (x,t) dx$ and the non-normalized density satisfies the Duncan-Mortensen-Zakai equation (Stratonovich form)
\begin{eqnarray*}
d \sigma = \mathcal{L}_0^\star \sigma \, dt + h(x) \sigma \circ dY (t)
\end{eqnarray*}
Here we encounter the dual generator
\begin{eqnarray*}
\mathcal{L}_0^\star = \frac{1}{2} \gamma^2_0 \triangle - \nabla . ( v \cdot ) - \frac{1}{2} h(x)^2 .
\end{eqnarray*}
In order to study the algebraic features of the filter, Brockett \cite{Brockett80} and Mitter \cite{Mitter79,Mitter80} independently introduced the notion of an \textit{estimation algebra} defined to be the Lie algebra
\begin{eqnarray*} 
\mathcal{E} = \mathrm{Lie} \big\{ \mathcal{L}_0^\star , h \big\}
\end{eqnarray*}
where $h$ is understood as the operator of multiplication by the function $h$ pointwise.
For instance, if $\mathcal{E}$ is has a finite basis $X_1, \cdots , X_m$, then one may hope to find a solution locally around $t>0$ of the Wei-Norman form
\begin{eqnarray*}
\sigma (x,t) = e^{ u_1 (t) X_1} \cdots e^{u_n (t) X_n}  \sigma (x,0) ,
\end{eqnarray*}
with further hope that this may be global if the estimation algebra is solvable. 

We remark that one may introduce a form of gauge covariant derivative $D= \nabla - \frac{1}{\gamma_0^2} v(x)$ which allows us to complete the squares so that
\begin{eqnarray*}
\mathcal{L}_0^\star = \frac{1}{2} \gamma^2_0 D^\top . D - \Phi (x) ,
\end{eqnarray*}
where $ \Phi (x) = \frac{1}{2} \big[ h(x)^2 + \nabla . v (x) + \frac{1}{\gamma_0^2}  v(x)^\top v(x) \big]$.
The derivative $D$ has a gauge field given by
\begin{eqnarray*}
F_{ij} = [D_i , D_j ] = \frac{\partial v_i}{ \partial x_j}-
\frac{\partial v_j}{ \partial x_i}.
\end{eqnarray*}
The filtering problem is said to be exact when $ F \equiv 0$ and this implies that the drift $v$ is gradient in $\mathbb{R}^n$. The Bene\v{s} filters, for example, are the family of finite filters associated with an exact problem for which the function $h(x)$ is linear and $\Phi (x)$ is at most quadratic. Brockett and Mitter also initiated the study of models for which the gauge field is constant and non-zero. For more details, see the excellent review \cite{Wong_Ya}.

One of our aims is to find a suitable quantum analogue to the estimation algebra for the quantum filters associated with homodyne measurements. Of necessity, we deal with super-operators in place of vectors, and here we make some observations and proposals.

First of all, tangent vector fields are the derivations on the algebra of smooth functions. If we replace this with the algebra, $ \mathfrak{B} ( \mathfrak{h} )$, of bounded operators on a Hilbert space then we encounter the following.
For a super-operator, $\mathcal{L}$, its dissipation is the map $\mathcal{D}_{\mathcal{L}}$ from 
$ \mathfrak{B}\left( \mathfrak{h}\right) \times \mathfrak{B}\left( \mathfrak{h}\right)$ to $\mathfrak{B}\left( \mathfrak{h}\right)$ given by
\begin{eqnarray*}
\mathcal{D}_{\mathcal{L}} (X,Y) = \mathcal{L} (XY) - \mathcal{L} (X)Y -X \mathcal{L} (Y).
\end{eqnarray*}
The derivations are the super-operators for which the dissipation vanishes identically.
A well-known Theorem of Sakai, \cite{Sakai}, tells us that the derivations on a C*-algebra are precisely the super-operators of the form $ -i [ \cdot, H] $ for some $H\in \mathfrak{B} ( \mathfrak{h} )$, and for a *-map we require that the $H$ be self-adjoint. The Lie bracket of two derivations with Hamiltonians $H_1$ and $H_2$ is again a derivation with Hamiltonian $ \frac{1}{i} [H_1 ,H_2 ]$.

One may be suspicious that the quantum derivations are not the most general analogue of tangent vectors - rather they are somehow only the analogues of Hamiltonian vectors fields on some manifold with a Poisson Brackets structure.
However, there is a larger class of super-operators than just the derivation *-maps that possesses good Lie algebraic closure properties and may be a more natural analogue of a vector fields. We refer to these a $\zeta$-type super-operators and they are defined as maps of the form
\begin{eqnarray}
\zeta_K (X) =K^\ast X +XK .
\end{eqnarray}
These generate the semigroups $e^{K^\ast t} X e^{Kt}$. The key Lie identity that motivates us is
\begin{eqnarray}
[ \zeta _{K_1} , \zeta_{K_2} ] =- \zeta_{ [K_1 , K_2  ]}.
\end{eqnarray}
The Lindblad generators arise as the naturally quantum analogue of 2nd order generators of Markov diffusions, and the pure form is
\begin{eqnarray}
\mathcal{L} X &=& \frac{1}{2} L^\ast [X,L] + \frac{1}{2} [L^\ast ,X ] L -i [X,H] \notag \\
&=&  L^\ast XL - \zeta_{ \frac{1}{2} L^\ast L + iH} (X) .
\end{eqnarray}
Traditionally, one considers the first form as the splitting into a 2nd order term (the $L,L^\ast$ parts) and a 1st order derivation (the $H$ part). However, the second form encourages us to split into a decoherent term ($L^\ast XL$) and a coherent term (the $\zeta$ part).

\section{The Belavkin-Zakai Equation}

\subsection{The Filter Equations}
Let us fix a Hilbert space $\mathfrak{h}$ and denote by $\mathfrak{B}\left( \mathfrak{h}\right) $ the C*-algebra of bounded operators on $\mathfrak{h}$. 
We consider a QSDE model with coupling operators $G=\left( \mathbf{L},H\right) $ where $\mathbf{L} = [L_1 , \cdots , L_n]$ is a column vector of operators on $\mathfrak{B} (\mathfrak{h})$ and $H$ is self-adjoint \cite{HP84} -\cite{Par92}:
\begin{eqnarray}
dU_G\left( t\right) =\left\{ \sum_k L_k \otimes dB_k^\ast \left( t\right) -L^{\ast }_k \otimes
dB_k \left( t\right) +K\otimes dt\right\} U_G \left( t\right) 
\end{eqnarray}
where 
\begin{eqnarray}
K=-\frac{1}{2}\sum_k L_k^{\ast }L_k-iH
\end{eqnarray}
and we initialize with $U\left( 0\right) =I$.

Let $A$ be a subset of $\{ 1, \cdots ,n \}$ then for each index $\alpha \in A$ we measure the output quadrature $Y_\alpha\left( t\right) =U_G\left( t\right)
^{\ast }\left[ I\otimes \left( e^{i\theta_\alpha }B_\alpha\left( t\right) +e^{-i\theta_\alpha
}B\left( t\right) \right) \right] U_G\left( t\right) $ with
\begin{eqnarray}
dY_\alpha\left( t\right) =e^{i\theta_\alpha }dB_\alpha\left( t\right) +e^{-i\theta_\alpha }dB_\alpha\left(
t\right) ^{\ast }+j_{t}\left( L_\alpha e^{i\theta_\alpha }+L_\alpha^{\ast }e^{-i\theta_\alpha }\right) dt,
\end{eqnarray}
and denote the set of quadrature phases as
\begin{eqnarray}
\mathbf{ \Theta } = \{ \theta_\alpha : \alpha \in A\}.
\end{eqnarray}

The filtered estimate, $\pi_t (X)$, of any system observable, $X \in \mathfrak{B} (\mathfrak{h})$ at time $t$ is defined to be the conditional expectation of $U_G(t)^\ast [X \otimes I ] U_G (t)$ onto the von Neumann algebra generated by the quadratures $\{ Y_\alpha (\tau ) : 0 \le \tau \le t, \alpha \in A \}$ and may be written as  
\begin{eqnarray}
\pi _{t}\left( X\right) =\frac{\sigma _{t}\left( X\right) }{\sigma
_{t}\left( I\right) }
\end{eqnarray}
where $\sigma _{t}\left( X\right) $ satisfies the Belavkin-Zakai equation (see \cite{Belavkin}- \cite{BvHJ} and in particular \cite{BvH_ref})
\begin{eqnarray}
d\sigma _{t}\left( X\right) =\sigma _{t}\left( \mathcal{L}_G X\right) dt+
\sum_\alpha \sigma
_{t}\left( XL_\alpha e^{i\theta_\alpha }+e^{-i\theta_\alpha }L_\alpha X\right) dY_\alpha \left( t\right) ,
\label{eq:BZ}
\end{eqnarray}
with $\sigma _{0}\left( X\right) =tr\left\{ \rho _{0}X\right\} $, where $%
\rho _{0}$ is the initial state of the system. Here, the dynamical part of the filter involves the associated Lindbladian
\begin{eqnarray}
\mathcal{L}_G X=\sum_k \bigg\{ L_k^{\ast
}XL_k-\frac{1}{2}L_k^{\ast }L_kX-\frac{1}{2}XL_k^{\ast }L_k \bigg\}
-i\left[ X,H\right] .
\end{eqnarray}

For this case, when the initial state of the system is a pure state ($\rho _{0}\equiv |\psi _{0}\rangle \langle \psi
_{0}|$) we may write the filter as $\sigma _{t}\left( X\right) \equiv \langle \chi
_{t}|X\chi _{t}\rangle $ where the non-normalized conditional state vector, $ | \chi_t \rangle$,
satisfies 
\begin{eqnarray}
d|\chi _{t}\rangle =K|\chi _{t}\rangle dt+ \sum_\alpha e^{i\theta_\alpha }L_\alpha |\chi _{t}\rangle
dY_\alpha \left( t\right) ,
\label{eq:BZ_pure}
\end{eqnarray}
with $|\chi _{0}\rangle =|\psi _{0}\rangle $.

\subsection{Lie Algebraic Formalism for Super-Operators}

A super-operator
is a linear map from $\mathfrak{B}\left( \mathfrak{h}\right) $ to itself and we may
endow it with a Lie bracket given by
\begin{eqnarray}
\left[ \mathcal{L}_{1},\mathcal{L}_{2}\right] \triangleq \mathcal{L}%
_{1}\circ \mathcal{L}_{2}-\mathcal{L}_{2}\circ \mathcal{L}_{1},
\end{eqnarray}
where $\circ $ denotes composition. As the composition is associative, it follows that these Lie brackets automatically satisfies the Jacobi identity. For a collection of super-operators $\left\{ \mathcal{L}_{1},\cdots ,\mathcal{L}_{n}\right\} $, we write $\mathrm{%
Lie}\left\{ \mathcal{L}_{1},\cdots ,\mathcal{L}_{n}\right\} $ for the Lie
algebra of super-operators generated using this bracket. We say that a super-operator is a *-map if $\mathcal{L} (X^\ast ) = \mathcal{L} (X)^\ast$ for every $X$.

\begin{definition}
For each $A\in \mathfrak{B}\left( \mathfrak{h}\right) $ we define a super-operator 
$\zeta _{A}:\mathfrak{B}\left( \mathfrak{h}\right) \mapsto \mathfrak{B}\left( \mathfrak{h}\right) $ by
\begin{eqnarray}
 \zeta _{A}\left( X\right)  \triangleq XA+A^{\ast }X.
\end{eqnarray}
For $\mathfrak{A}$ a subset of $\mathfrak{B}\left( \mathfrak{h}\right) $, we write 
$\zeta _{\mathfrak{A}}$ for the set $\left\{ \zeta _{A}:A\in \mathfrak{A}\right\} $ of super-operators. 
\end{definition}

We note that if $A= A^\prime + i A^{\prime \prime}$ with $A^\prime$ and $A^{\prime \prime}$ self-adjoint, then
\begin{eqnarray}
\zeta_A  (X) = [ X,A^\prime ]_+ + i[X, A^{\prime \prime} ],
\end{eqnarray}
where the anti-commutator is $ [X,Y]_+ = XY+YX$.

The maps $\zeta_A$ have the following properties:

\begin{enumerate}
\item  $\zeta _{A}$ is a derivation if and only if $A$ is skew-adjoint,
indeed the dissipation is
\begin{eqnarray}
\mathcal{D}_{\zeta _{A}} \left( X,Y\right)  =-X\left( A+A^{\ast }\right) Y
\end{eqnarray}
for all $X,Y\in \mathfrak{B}\left( \mathfrak{h}\right) $;

\item  $\zeta _{A}$ is a *-map, that is, $\zeta _{A}\left( X\right) ^{\ast
}=\zeta _{A}\left( X^{\ast }\right) $;

\item  the map $:A\mapsto \zeta _{A}$ is also linear;

\item the adjoint of $\zeta_A$ (with respect to the dual set of trace-class operators: $\mathrm{tr}\{ \mathcal{L}^\star ( \rho ) X \}
= \mathrm{tr} \{ \rho \mathcal{L}(X) \}$) is 
\begin{eqnarray*}
\zeta_A^\star \equiv \zeta_{A^\ast} .
\end{eqnarray*}

\item we have
\begin{eqnarray}
\zeta _A \circ \zeta_A = 2A^\ast XA +XA^2 +A^{\ast 2} X.
\label{eq:comp_id}
\end{eqnarray}

\item the Lindblad generator may be written as
\begin{eqnarray}
\mathcal{L}_G (X)= \sum_k L^\ast_k X L_k + \zeta_K X .
\end{eqnarray}

\item we have that $L^\ast XL \equiv \frac{1}{2} \big( \zeta_L \circ \zeta_L - \zeta_{L^2} \big)$ and so the
Lindblad generator may also be written as
\begin{eqnarray}
\mathcal{L}_G (X)= \frac{1}{2} \sum_k \zeta_{L_k} \circ \zeta_{L_k} -  \zeta_{\frac{1}{2} \sum_k (L_k +L^\ast_k) L_k +iH} ( X),
\end{eqnarray}
or as
\begin{eqnarray}
\mathcal{L}_G = \frac{1}{2} \sum_k \bigg( \zeta_{L_k} \circ \zeta_{L_k} -  \zeta_{ L_k^2} ( X) \bigg) + \zeta_K.
\end{eqnarray}

\end{enumerate}

\begin{proposition}
We have the Lie bracket identity
\begin{eqnarray}
\left[ \zeta _{A},\zeta _{B}\right] =-\zeta _{\left[ A,B\right] }.
\label{eq:Lie_homomorphism}
\end{eqnarray}
\end{proposition}

\bigskip

The property (\ref{eq:Lie_homomorphism}) yields a Lie algebra homomorphism from the set of
operators $\mathfrak{B}\left( \mathfrak{h}\right) $ with the commutator as bracket to
the set of super-operators with the super-operator Lie bracket. Indeed, we
then have 
\begin{eqnarray}
\mathrm{Lie}\left\{ \zeta _{A_{1}},\cdots ,\zeta _{A_{n}}\right\} \equiv
\zeta _{\mathrm{Lie}\left\{ A_{1},\cdots ,A_{n}\right\} }.
\end{eqnarray}

The $\zeta$-type super-operators are not derivations, they do provide a nontrivial class of \lq\lq Lie vectors\rq\rq \,  by virtue of the identity (\ref{eq:Lie_homomorphism}).

\subsection{Stratonovich Form of The Belavkin-Zakai Equation}
In order to extract the Lie algebraic features of Belavkin-Zakai equation (\ref{eq:BZ}) we change to the Stratonovich form. 
The Stratonovich derivative is $X\circ dY=XdY+\frac{1}{2}dXdY$.

\begin{proposition}
The Stratonovich form of the Belavkin-Zakai equation (\ref{eq:BZ}) is
\begin{eqnarray}
d\sigma _{t}\left( X\right) =\sigma _{t}\left( \tilde{\mathcal{L} }_{G, \mathbf{\Theta}} (X) \right)
dt+ \sum_{\alpha \in A} 
\sigma _{t}\left( \zeta _{e^{i\theta_\alpha }L}X\right) \circ dY_\alpha \left( t\right) 
\label{eq:BZ_Strat}
\end{eqnarray}
where
\begin{eqnarray}
\tilde{\mathcal{L} }_{G, \mathbf{\Theta}} = \mathcal{L}_G -\frac{1}{2} \sum_{\alpha \in A} 
\zeta _{e^{i\theta_\alpha }L_\alpha }\circ \zeta _{e^{i\theta_\alpha }L_\alpha} .
\label{eq:L_tilde}
\end{eqnarray}
\end{proposition}

\begin{proof}
The Belavkin-Zakai equation may be written as 
\begin{eqnarray*}
d\sigma _{t}\left( X\right)
=\sigma _{t}\left( \mathcal{L}_G X\right) dt+\sum_\alpha \sigma _{t}\left( \zeta
_{L_\alpha e^{i\theta_\alpha }}X\right) dY_\alpha \left( t\right) 
\end{eqnarray*}
and we make the ansatz that its
Stratonovich form is $d\sigma _{t}\left( X\right) =\sigma _{t}\left( 
\tilde{\mathcal{L}}_{G, \mathbf{\Theta}}X\right) dt+ \sum_{\alpha \in A} \sigma _{t}\left( \zeta _{L_\alpha e^{i\theta_\alpha
}}X\right) \circ dY_\alpha \left( t\right) $. Converting to Ito form, we see that we must have
\begin{eqnarray*}
\mathcal{L}_G X &\equiv &\tilde{\mathcal{L}}_{G, \mathbf{\Theta}}X+\frac{1}{2} \sum_\alpha \zeta _{e^{i\theta_\alpha
}L_\alpha }\circ \zeta _{e^{i\theta_\alpha  }L_\alpha }X .
\end{eqnarray*}
from which we obtain the stated form for $\tilde{\mathcal{L}}_{G, \mathbf{\Theta}}$. 
\end{proof}

We note, using identity (\ref{eq:comp_id}), that 
\begin{eqnarray}
\tilde{\mathcal{L}}_{G, \mathbf{\Theta} } X &= &
\sum_k \bigg\{ L_k^{\ast
}XL_k-\frac{1}{2}L_k^{\ast }L_kX-\frac{1}{2}XL_k^{\ast }L_k \bigg\}
-i\left[ X,H\right]  \notag \\
&&- \frac{1}{2} \sum_{\alpha \in A} \bigg( 2 L^\ast_\alpha X L_\alpha + X L^2_\alpha e^{2i\theta_\alpha
} + L_\alpha^{\ast 2 } e^{-2i\theta_\alpha  }X  \bigg)  \notag \\
&=&
\zeta_{K(G,\mathbf{\Theta} )} + \mathcal{L}_{\text{unobs.}}
,
\end{eqnarray}
where the $\zeta$-term involves
\begin{eqnarray}
K (G, \mathbf{\Theta} )=-\frac{1}{2}\sum_{\alpha \in A} \bigg( L_\alpha^{\ast }L_\alpha - L_\alpha^{2}e^{2i\theta_\alpha }\bigg)
-iH,
\end{eqnarray}
and the Linbdladian associated with the unobserved, $ k \notin A$, channels is
\begin{eqnarray}
\mathcal{L}_{\text{unobs.}} =
\sum_{k \notin A} \bigg\{ L_k^{\ast
}XL_k-\frac{1}{2}L_k^{\ast }L_kX-\frac{1}{2}XL_k^{\ast }L_k \bigg\}
.
\end{eqnarray}

\begin{definition}
Given the open quantum system described by $G \sim (\mathbf{L} , H )$ and the quadrature measurement scheme determined by the subset
$A$ of outputs with quadrature phases $\mathbf{\Theta} $, we define the \textbf{Estimation Algebra} associated with the corresponding quantum filter to be
\begin{eqnarray}
\mathscr{L} \left( G, \mathbf{\Theta}  \right) \triangleq \mathrm{Lie}\left\{ \tilde{\mathcal{L}}_{G, \mathbf{\Theta}},\zeta _{e^{i\theta_\alpha }L_\alpha } : \alpha \in A \right\} 
\end{eqnarray}
with $\tilde{\mathcal{L}}_{G, \mathbf{\Theta} }$ as given by (\ref{eq:L_tilde}).
\end{definition}

The term estimation algebra is taken from Brockett \cite{Brockett80}. Note that dissipative super-operators exponentiate to give semi-groups only. Typically, we should consider only a cone inside the estimation algebra.

\section{Complete Homodyne Observations}
We shall refer to the special case where we  perform a homodyne measurement on \textit{all} the output channels, so that 
$A \equiv \{ 1 , \cdots, n \}$, as \textbf{complete homodyne} detection. Here we find that $\mathcal{L}_{\text{unobs.}} \equiv 0$ which leads to the remarkable fact that the Stratonovich operator $\tilde{ \mathcal{L}}_{G, \mathbf{\Theta}}$ is now a purely to a $\zeta$-type super-operator: 
\begin{eqnarray}
\tilde{\mathcal{L}} \equiv \zeta_{K(G, \mathbf{\Theta})} \qquad (\textbf{complete homodyne})
\end{eqnarray}
where
\begin{eqnarray}
K (G, \mathbf{\Theta} )=-\frac{1}{2}\sum_k \bigg( L_k^{\ast }L_k -\frac{1}{2}L_k^{2}e^{2i\theta_k } \bigg) -iH.
\label{eq:K_complete}
\end{eqnarray}
Note that here the Stratonovich form (\ref{eq:BZ_Strat}) of the Belavkin-Zakai equation involves only
the super-operators $\zeta _{K }$ and the $\zeta _{e^{i\theta }L}$ as
the term $\sum_k L_k^{\ast }XL_k$ cancels exactly from the Lindbladian irrespective of the values of the phases. 
Specifically we have 
\begin{eqnarray}
d \sigma_t (X) &=& \sigma_t ( \zeta_{K ( G, \mathbf{\Theta} )} X )\, dt
+\sum_k \sigma_t (\zeta_{ e^{i \theta_k} L_k } X ) \circ dY_k (t) \notag \\
&& \qquad \qquad \qquad \qquad \qquad
\qquad (\textbf{complete homodyne}).
\end{eqnarray}

\begin{proposition}
\label{prop:L_special}
If $A \equiv \{ 1 , \cdots, n \}$, then 
the pure state equation (\ref{eq:BZ_pure}) has the Stratonovich form
\begin{eqnarray}
 d|\chi _{t}\rangle =\left( K(G, \mathbf{\Theta} ) \, dt+\sum_k e^{i\theta_k }L_k dY_k \left( t\right)
\right) \circ |\chi _{t}\rangle .
\label{eq:BZ_pur_Strat}
\end{eqnarray}
\end{proposition}

As the Stratonovich calculus observes the Leibniz rule, we see that the corresponding Belavkin-Zakai equation
will involve only $\zeta $-type super-operators and no It\={o} correction, in agreement with Proposition \ref{prop:L_special}.

(We remark that the same class of equations as (\ref{eq:BZ_pur_Strat}) have been studied by Rebolledo, Mora and Fagnola, \cite{Mora_Rebolledo} and \cite{Fagnola_Mora_2015}, but with the $Y_k$ taken as Wiener processes. In \cite{Fagnola_Mora_2015}, they obtain an interesting Lie algebra rank condition for controllability of the equations when the $Y_k$ are replaced by piecewise continuous functions and relate this to irreducibility of the corresponding quantum dynamical semi-group generated by the Lindbladian.)

Putting together the fact that we only have $\zeta$-type super-operators in the filter and the fact that we have the Lie-algebra homomorphism (\ref{eq:Lie_homomorphism}), we obtain the following observation.

\begin{theorem}
\label{thm:main}
Given an open quantum system determined by $G = (\mathbf{L}, H)$ where we perform a complete homodyne observations measurement (on quadratures of all the output channels), the estimation Lie algebra associated with the filter is 
\begin{eqnarray}
\mathscr{L} \left( G,\theta \right) \equiv \zeta_{\mathfrak{L}(G, \mathbf{\Theta} )} \qquad (\textbf{complete homodyne})
,
\end{eqnarray}
where 
\begin{eqnarray}
\mathfrak{L} (G, \mathbf{\Theta} )
 \triangleq \mathrm{Lie}\left\{ K(G, \mathbf{\Theta} ) , e^{i\theta_1} L_1 , \cdots, e^{i\theta_n} L_n \right\},
\end{eqnarray}
with given by (\ref{eq:K_complete}).
\end{theorem}

\section{Conclusion}
We introduced the concept of an estimation Lie algebra for quantum homodyne detection problems.
We have restricted our interest to homodyne detection where the all the noise channels are measured and where the measurements are ideal. However, in this case we see that the quantum filter in the Stratonovich calculus is described entirely through the class of $\zeta$-superoperators introduced. This class possesses a nice closure property for the Lie bracket of commutators of superoperators and through the Lie-algebra homomorphism (\ref{eq:Lie_homomorphism}) we obtain a simple form for the corresponding estimation Lie algebra for the filter. Clearly a necessary condition for the filter to be finite is that the associated Lie algebra $\mathfrak{L} (G, \mathbf{\Theta} )$ appearing in Theorem \ref{thm:main} be finite. We will look at applications in a future publication, however, believe that the estimation algebra of the quantum filter will be a useful conceptual device for studies.

\bigskip

\textbf{Acknowledgements}
We wish to thank the Institut Henri Poincar\'{e} for their kind support during the IHP Trimester \textit{Measurement and Control of Quantum Systems: Theory and Experiments}, Paris, France, 2018.

\end{document}